\documentclass[letterpaper, 10 pt, conference]{ieeeconf}  % Comment this line out
\IEEEoverridecommandlockouts                              % This command is only
\usepackage{graphicx} % for pdf, bitmapped graphics files
\DeclareGraphicsExtensions{.jpg,.png,.pdf}
\usepackage{amsmath} % assumes amsmath package installed
\usepackage{amssymb}
\usepackage{framed}
\usepackage{cite}        % required for bibliography
\newtheorem{theorem}{Theorem}
\newtheorem{lemma}{Lemma}
\newtheorem{proposition}{Proposition}

\newtheorem{assumption}{Assumption}
\newtheorem{remark}{Remark}
\title{\LARGE \bf
Optimal Distributed Controller Design with\\
Communication Delays: Application to Vehicle Formations
}

\author{Hamid Reza Feyzmahdavian, Assad Alam, and Ather Gattami% <-this % stops a space
\thanks{H. R. Feyzmahdavian is with Electrical Engineering, ACCESS Linnaeus Centre, Royal Institute of Technology, 100 44
Stockholm, Sweden.
        {\tt\small hamidrez@kth.se}}%
\thanks{A. Alam is with Research and Development, Scania CV AB, 151 87 S\"{o}dert\"{a}lje, Sweden.
        {\tt\small assad.alam@scania.com}}%
\thanks{A. Gattami is with Electrical Engineering, ACCESS Linnaeus Centre, Royal Institute of Technology, 100 44
Stockholm, Sweden.
        {\tt\small gattami@kth.se}}%
\thanks{This work was supported by Scania CV AB, VINNOVA - FFI, and the Swedish Research Council.}% <-this % stops a space
}

\begin{document}

\maketitle
\thispagestyle{empty}
\pagestyle {empty}

%%%%%%%%%%%%%%%%%%%%%%%%%%%%%%%%%%%%%%%%%%%%%%%%%%%%%%%%%%%%%%%%%%%%%%%%%%%%%%%%
\begin{abstract}

This paper develops a controller synthesis algorithm for distributed LQG control problems under output
feedback. We consider a system consisting of three interconnected linear subsystems with a delayed information sharing structure. While the state-feedback case of this problem has previously been solved, the extension to output-feedback is nontrivial, as the classical separation principle fails. To find the optimal solution, the controller is decomposed into two independent components. One is delayed centralized LQR, and the other is the sum of correction terms based on additional local information. Explicit discrete-time equations are derived whose solutions are the gains of the optimal controller.\footnote{A preliminary version of this work was presented in~\cite{Feyzmahdavian:12}.}

\iffalse We consider distributed control design for chain structured systems with communication delays. An analytical expression is derived for an optimal distributed controller over the class of LQG control, which accounts for the interconnection between neighboring vehicles as well as communication delays. We present an approach that decomposes the states into orthogonal components, such that the cost function can be decomposed into separate optimization problems. The results show that a tight control is achieved, even in the presence of delays, and it behaves well with respect to the imposed disturbances. The proposed optimal controller outperforms a centralized controller with two time step delays and it has 0.01\,\% higher computed theoretical cost with respect to a centralized controller without delay.\fi

%is significantly better than the theoretical cost for a centralized controller %with two time step delays. Furthermore, the computed theoretical cost is close %to a fully centralized controller with full state information at all times.

\end{abstract}

%%%%%%%%%%%%%%%%%%%%%%%%%%%%%%%%%%%%%%%%%%%%%%%%%%%%%%%%%%%%%%%%%%%%%%%%%%%%%%%%
\section{INTRODUCTION}

The systems to be controlled are in many application domains getting larger and more complex. When there is interconnection between different dynamical systems, conventional optimal control algorithms provide a solution where centralized state information is required. However, it is often preferable and sometimes necessary to have a distributed control structure, since in many practical problems, the physical or communication constraints often impose a specific interconnection structure. Hence, it is interesting to design distributed feedback controls for systems of a certain structure and examine their overall performance.

The control problem and methodology in this paper is motivated by systems involving a chain of closely spaced heavy duty vehicles (HDVs), generally referred to as vehicle platooning. The objective is to maintain a predefined headway to the vehicle ahead, while maintaining safety and minimizing the fuel consumption. Information technology is paving its path into the transport industry, enabling the possibility of automated control strategies. Governing vehicle platoons by an automated control strategy, the overall traffic flow is expected to improve \cite{IoannouChien93} and the road capacity will increase significantly \cite{DeSchutter:99}, without endangering safety \cite{Alam11}. By traveling at a close intermediate spacing the air drag is reduced for each vehicle in the platoon. Thereby, the control effort and inherently the fuel consumption can be reduced significantly \cite{Alam10}. This creates a coupling of the dynamics between neighboring vehicles throughout the platoon. However, as the intermediate spacing is reduced the control becomes tighter due to safety aspects; mandating an increase in control action and inherently the fuel consumption through additional acceleration and braking. The fuel consumption constitutes approximately 30\,\% of the overall cost for a fleet owner \cite{Alam:Lic}. Hence, it is of vast interest for the industry to find a fuel optimal control. Considering the physical constraints in radio, it cannot be assumed that state information is available at every instance in time. Thus, a distributed control strategy is crucial for practical implementation.

In recent work \cite{Bamieh02,DAndrea98,Feyzmahdavian:12-1}, distributed control has been studied under the assumption of spatial invariance. \iffalse Decentralized control design for stability based on local model knowledge appeared in \cite{Khorsand:11}.\fi  ~Control for chain structures in the context of platoons has been studied through various perspectives, e.g., \cite{Bamieh08,Barooah05,BamiehJovanovic05,Hedrick96,Varaiya93}. It has been shown that control strategies may vary depending on the available information within the platoon. Maintaining a suitable relative distance, stability and robustness of the platoon have been identified to be amongst the main criteria to be considered. However, communication constraints have not in general been considered in control design for platooning applications and the controllers have mainly been ad hoc by tuning the control parameters. In \cite{rantzer:acc06,gattami:06}, linear quadratic Gaussian (LQG) control under appropriate assumptions on communication delays between the controllers was considered. While a computationally efficient solution was presented for a sequence of vehicles moving in formation, the controller structure is not provided by the corresponding semi-definite programming. \iffalse In \cite{lamperski:cdc11}, an optimal LQG control for the three-player chain was solved, but disturbances acting on players were assumed to be independent.\fi A structured sequential design was introduced in \cite{Alam11:Dec}, where the preceding vehicle's dynamics along with its states were conveyed through wireless communication. It resulted in a suboptimal control strategy, where physical coupling to a follower vehicle and communication delays were not considered. Mounted radar sensors allows each vehicle to measure the relative distance and velocity of the preceding vehicle. Additional information, providing local information, has lately been introduced through wireless information. However, wireless systems introduce information delays to the system in certain cases due to limitations in radio. Furthermore, varying external environment factors impose process disturbances on the system.

In this work, we are primarily concerned with forming a distributed control, that accounts for the interconnection between neighboring vehicles, correlated process disturbances, as well as communication delays. The control is solely based on local model knowledge, over the class of LQG control for chain structured interconnection graphs. The received information is assumed to be common after two time step delays.

The main contribution of this paper is to derive an LQG controller, which is easy to implement and optimal under a delayed information sharing pattern for chain structures.  In addition to communication delays, the distributed optimal control is based upon systems with interconnected dynamics to both neighboring vehicles and local state information. Derived from the characteristics of actual Scania HDV's, we present a discrete system model that includes physical coupling with both neighboring vehicles. We also investigate the performance of the proposed controllers, under normal operating conditions for an HDV platoon, with respect to physical constraints that are imposed in a practical set-up.

The outline of the paper is as follows. The general system and problem description is given in Sec.~\ref{sec:Problem}, which in turn determines the structure of the optimal controller. The theoretical premise for the optimal controller is presented in Sec.~\ref{sec:OptimalController}, where it is shown that the problem can be decomposed into two separate optimization problems. Finally, we evaluate the performance of the derived controller through numerical results in Sec.~\ref{sec:Performance} and give concluding remarks in Sec.~\ref{sec:Concl}.

\subsection*{Notation}
\label{sec:preliminaries}
Throughout the paper, we use the following notation: matrices are written in uppercase letters and vectors in lowercase letters.  The $i^{th}$ component of a vector $x$ is denoted by $x_i$. Let
$[x]_S$ be the sub-vector of $x$ containing only those components with indices in set $S$. For instance, if $S=\{1,3\}$, then $[x]_S$ is given by $[x]_S=\begin{bmatrix} x_1 & x_3\end{bmatrix}^T$. The sequence $x(0)$, $x(1)$, $\ldots$ , $x(k)$ is denoted by $x(0:k)$.

\text{diag}$(x)$ denotes a diagonal matrix whose diagonal elements are given by those of the enclosed vector $x$. Let $X$ be a matrix partitioned into blocks. We use $[X]_{ij}$ and $[X]_i$ to represent the block in block position $ij$ and $i$th block row, respectively. $[X]_{S_1S_2}$ denotes
the sub-matrix of $X$ containing exactly those rows and columns corresponding to the sets $S1$ and $S2$, respectively. For instance
$[X]_{\{1\}\{2,3\}}=\begin{bmatrix} X_{12} & X_{13}\end{bmatrix}$. The trace of square matrix is denoted by $\textbf{Tr}\{X\}$. We use $X^+$ and $X^-$ to represent $X(k+1)$ and $X(k-1)$ respectively, when appropriate.

Given $A\in\mathbb{R}^{m\times n}$, we can write $A$ in terms of its columns as
$A=\begin{bmatrix} a_1 & \cdots & a_n\end{bmatrix}$. The operation \textrm{vec}$(A)$ results in a $mn \times 1$ column vector
$\textrm{vec}(A)=\begin{bmatrix} a_1^T & \cdots & a_n^T\end{bmatrix}^T$. We denote by $\textrm{vec}^{\star}(A)$, the sub-vector of $\textrm{vec}(A)$ containing only nonzero elements. Let $A\in\mathbb{R}^{m\times n}$ and $B\in\mathbb{R}^{r\times s}$, then the operation $A\otimes B \in\mathbb{R}^{mr\times ns}$ denotes the \textit{Kronecker product} of $A$ and $B$.

We denote the expectation of a random variable $x$ by $\textbf{E}\{x\}$. The conditional expectation of $x$ given $y$ is denoted by $\textbf{E}\{x|y\}$.

\section{System Model and Problem Description}
\label{sec:Problem}
In this section we present the physical properties of the system that we are considering. We state the nonlinear dynamics of a single vehicle and the model for the aerodynamics, which induces the physical coupling. Then we present the linear discrete system model for a heterogeneous HDV platoon and its associated cost function. The communication constraints and physical coupling is then used to motivate the structure of the controller. Finally, the problem formulation is given.

\subsection{System Model}
We consider an HDV platoon as depicted in Fig.~\ref{fig:platoon}. The state equation of a single HDV is modeled as,

\begin{equation}
	\begin{aligned}
	%\begin{split}
 		\dot{s} 		&= v,\\
		m_t\dot{v}	&=F_{engine} -F_{brake}-F_{air drag}(v)\\
 									& \hspace{15.6mm}				 -F_{roll}(\alpha)-F_{gravity}(\alpha), \\
 									&=k_uu -k_bF_{brake}-k_dv^2 \\
 									& \hspace{11.6mm}-k_{fr}\cos\alpha-k_g\sin\alpha,
		%\end{split}
	\end{aligned}
\label{eq:Forces}
\end{equation}

\noindent where $v$ is the vehicle velocity, $m_t$ denotes the accelerated mass and $u\in \mathbb{R}$ denotes the net engine torque. $k_u, k_b, k_d, k_{fr}$, and $k_g$ denote the characteristic vehicle and environment coefficients for the engine, brake, air drag, road friction, and gravitation respectively.

\begin{figure}[t]
\begin{center}
\includegraphics[width=8.4cm]{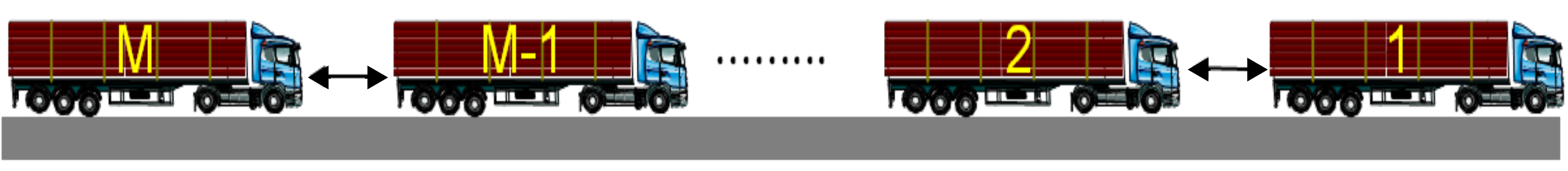}    % The printed column width is 8.4 cm.
\caption{The figure shows a platoon of $M$ heavy duty vehicles, where each vehicle is able to communicate with its neighbors.}
\label{fig:platoon}
\end{center}
\end{figure}

The variation in aerodynamics between the vehicles is essential in the analysis of fuel reduction potential for HDVs. For a single HDV it can amount up to 50\,\% of the total resistive forces at full speed. It is significantly reduced when operating in a platoon formation and a coupling between the vehicles is induced. To account for the aerodynamics, the air drag characteristic coefficient in \eqref{eq:Forces} can be modeled as \cite{Hucho}
\begin{align*}
\begin{split}
\tilde{k}_d&=k_d(1-\frac{\Phi(d)}{100}-\frac{\phi(d)}{100}),\\
\Phi(d)&=\alpha_1d + \alpha_2, 0 \leq d \leq 60\\
\phi(d)&=\beta_1d + \beta_2, 0 \leq d \leq 15
\end{split}
\end{align*}
\noindent where $d$ is the longitudinal relative distance between two vehicles, $\Phi(d)$ and $\phi(d)$ are linear piecewise affine functions of the change in air drag due to a preceding and a follower vehicle respectively, and $\alpha_1, \alpha_2, \beta_1, \beta_2$ are positive constants. The relative distance reference could be constant or, as in this case, time varying. It is determined by setting a desired time gap $\tau$\,s, which in turn determines the spacing policy as $d_{ref}(k)={\tau}v(k)$. Thereby, the vehicles will maintain a larger intermediate spacing at higher velocities.

When studying the behavior of an HDV platoon, the velocity does not deviate significantly from the lead vehicle's velocity. Ideally, all vehicles should maintain a constant speed and intermediate distance. Thus, a linearized model should give a sufficient description of the system behavior under these conditions. By linearizing and applying a one step forward discretization to \eqref{eq:Forces}, the discrete model for an HDV platoon with respect to a set reference velocity, an engine torque which maintains the velocity, a fixed spacing between the vehicles, and a constant slope is hence given by

\begin{equation*}
x(k+1) = Ax(k) + Bu(k)+w(k),
\end{equation*}

\noindent where
\begin{small}
\begin{align*}
A&= \begin{bmatrix}{\Theta_1} & \gamma_2 & 0 & 0 & 0 & \cdots & 0 & 0 & 0 \\
1 & 1 & -1 & 0 & 0 & \cdots & 0 & 0 & 0 \\
0 & \delta_2 & \Theta_2 & \gamma_3 & 0 & \cdots & 0 & 0 & 0 \\
0 & 0 & 1 & 1 & -1 & \cdots & 0 & 0 & 0 \\
0 & 0 & 0 & \delta_3 & \Theta_3 & \cdots & \gamma_4 & 0 & 0 \\
\vdots & \vdots & \vdots & \vdots & \vdots & \ddots & \vdots & \vdots & \vdots \\
0 & 0 & 0 & 0 & 0 & \cdots & \Theta_{M-1} & \gamma_{M-1} & 0 \\
0 & 0 & 0 & 0 & 0 & \cdots & 1 & 1 & -1 \\
0 & 0 & 0 & 0 & 0 & \cdots & 0 & \delta_M & \Theta_M \\
\end{bmatrix},
\end{align*}
\begin{align}
\begin{split}
B&= \begin{bmatrix} k_{u_1} & 0 & 0 & \cdots & 0 \\
0 & 0 & 0 & \cdots & 0 \\
0 & k_{u_2} & 0 & \cdots & 0 \\
0 & 0 & 0 & \cdots & 0 \\
0 & 0 & k_{u_3} & \cdots & 0 \\
\vdots & \vdots & \vdots & \ddots & \vdots \\
0 & 0 & 0 & \cdots & 0 \\
0 & 0 & 0 & \cdots & k_{u_M} \end{bmatrix}, \quad x=\begin{bmatrix}  {v_1} \\
{d_{12}}  \\
{v_2} \\
{d_{23}} \\
{v_3} \\
\vdots \\
{v_{M-1}} \\
{d_{(M-1)M}} \\
{v_M} \end{bmatrix}, \\
u&=\begin{bmatrix}  {u_1} \\
{u_2}  \\
{u_3} \\
\vdots \\
{u_M} \end{bmatrix}, \begin{array}{rl}
\Theta_1 & =1-T_s2k_d(d_0)v_0/m_t,\\
\Theta_i & =1-T_s2k_d\Phi{(d_0)}v_0/m_t, \quad i=2,\dots, M,\\
\delta_i & =-T_s\alpha_1k_dv_0^2/m_t,\\
\gamma_i & =-T_s\beta_1k_dv_0^2/m_t,
\end{array}
\end{split}
\label{eq:plantModel}
\end{align}
\end{small}
\noindent and $T_s$ is the sampling time. Thus, the system has a block diagonal structure and can be grouped into subsystems as indicated in \eqref{eq:plantModel}. The general representation of the derived system can be stated as

\begin{small}
\begin{align}
\hspace{-0.2cm}\begin{bmatrix} x_{1}(k+1) \\
x_{2}(k+1) \\
x_{3}(k+1) \\
\vdots \\
x_{M}(k+1)
 \end {bmatrix}=&\begin {bmatrix}A_{11} &  A_{12} & 0 & \cdots & 0\\
 A_{21} & A_{22} & A_{23} & \cdots & 0\\
 0 & A_{32} & A_{33}  & \cdots & 0\\
\vdots & \vdots & \vdots & \ddots & \vdots \\
0 & 0 & 0  & \cdots & A_{MM}
\end {bmatrix}\begin{bmatrix} x_{1}(k) \\
 x_{2}(k) \\
 x_{3}(k) \\
\vdots \\
x_{M}(k)
 \end {bmatrix}\nonumber\\
+&\begin{bmatrix} B_{1} &  0 &  0 & \cdots & 0\\
 0 & B_{2} &  0 & \cdots & 0\\
 0 &  0 &  B_{3} & \cdots & 0\\
 \vdots &  \vdots &  \vdots & \ddots & \vdots\\
 0 &  0 &  0 & \cdots & B_{M}\\
 \end {bmatrix}
\begin{bmatrix} u_{1}(k) \\
 u_{2}(k) \\
 u_{3}(k)\\
\vdots \\
u_{M}(k)
 \end {bmatrix}+w(k)
\label{mainsystem}
\end{align}
\end{small}

\noindent where the corresponding vehicle states for each subsystem are
\begin{equation*}
x_1(k)=v_1(k),\quad x_i(k)=\begin{bmatrix}d_{i-1,i} \\ v_i
\end{bmatrix},\quad i=2,\dots,M.
\end{equation*}

In practice, many random disturbances are imposed upon a vehicle in motion. The varying road topology has a strong impact due the extensive mass of the HDVs. Weather conditions might vary and traffic conditions might change. Furthermore, variation in wind affects all the vehicles in the platoon and therefore the process noise is considered to be correlated. Hence, the disturbance, $w(k)$ in \eqref{mainsystem}, is assumed to be a Gaussian white noise with a full positive definite covariance matrix $W$. We also assume that the initial state $x(0)$ is uncorrelated with $w(k)$ for all $k$, with zero mean and covariance matrix $P_0$.

While a general problem was defined, for simplicity, consider an $M=3$ HDV platoon. In this case, the dynamics of the system given in \eqref{eq:3dynamics} is
\begin{small}
\begin{align}
x_{1}(k+1)&=A_{11}x_{1}(k)+A_{12}x_{2}(k)+B_1u_1(k)\nonumber\\
x_{2}(k+1)&=A_{21}x_{1}(k)+A_{22}x_{2}(k)+A_{23}x_{3}(k)+B_2u_2(k)\nonumber\\
x_{3}(k+1)&=A_{32}x_{2}(k)+A_{33}x_{3}(k)+B_3u_3(k)
\label{eq:3dynamics}
\end{align}\end{small}
It can be seen in \eqref{eq:3dynamics} that the state of vehicle $1$ is affected by the states of vehicle $2$ in the next time step. Whereas, the state of vehicle $1$ affects the states of $3$ after two time steps, through vehicle $2$. Vehicle $2$ on the other hand is affected by both vehicle $1$ and $3$ in the next time step.

The local models can be conveyed at a single point in time between each subsystem, through wireless communication. However, the system is time critical due to safety aspects and communication should be kept at minimum so the channel is not congested and latency is introduced. Assume that passing information from one vehicle to another vehicle takes one time step, so the available information set of each vehicle at time $k$ can be described as
\begin{small}
\begin{align}
\mathcal{I}_{1}(k)&=\{x_1(k),x_1(k-1),x_2(k-1),x(0:k-2)\}\nonumber\\
\mathcal{I}_{2}(k)&=\{x_2(k),x(k-1),x(0:k-2)\}\nonumber\\
\mathcal{I}_{3}(k)&=\{x_3(k),x_2(k-1),x_3(k-1),x(0:k-2)\}
\label{historyset}
\end{align}\end{small}
The three vehicles share all past information with two-step communication delay, as described in \eqref{historyset}. The assumptions about the information structure and the sparsity of dynamics guarantee that information propagates at least as fast as the dynamics. This information pattern is a  simple case of \textit{partially nested} information structure. It is shown in \cite {Chi:72} that if the information structure is partially nested, then the optimal controller exists, it is unique, and linear. Therefore, the optimal controller for three vehicles under the given information set has the form
\begin{small}\begin{eqnarray}
u_1(k)&=&f_{11}\bigl(x_1(k)\bigr)+f_{12}\bigl(x_1(k-1),x_2(k-1)\bigr)\nonumber\\
&&\hspace{1.73cm}+f_{13}\bigl(x(0:k-2)\bigr)\nonumber\\
u_2(k)&=&f_{21}\bigl(x_2(k)\bigr)+f_{22}\bigl(x(k-1)\bigr)\nonumber\\
&&\hspace{1.73cm}+f_{23}\bigl(x(0:k-2)\bigr)\nonumber\\
u_3(k)&=&f_{31}\bigl(x_3(k)\bigr)+f_{32}\bigl(x_2(k-1),x_3(k-1)\bigr)\nonumber\\
&&\hspace{1.73cm}+f_{33}\bigl(x(0:k-2)\bigr)
\label{uc}
\end{eqnarray}\end{small}
where $f_{ij}$ denotes a linear function in all its variables. Consequently, the optimal control $u(k)$ can be expressed as
\begin{small}\begin{eqnarray}
u(k)=F(k) x(k)+G(k){x(k-1)}+f\bigl({x}(0:k-2)\bigr)
\label{ucontrol}
\end{eqnarray}\end{small}
where $f={\begin{bmatrix}  f_{13}^T & f_{23}^T & f_{33}^T \\ \end{bmatrix}}^T$ and
\begin{align*}
F(k)&=\begin{bmatrix} F_{11} &  0 &  0 \\  0 & F_{22} &  0 \\  0 & 0 &  F_{33}\\ \end{bmatrix},\;
G(k)&=\begin{bmatrix} G_{11} &  G_{12} & 0 \\  G_{21} &  G_{22} &  G_{23} \\  0 & G_{32} & G_{33}\\ \end{bmatrix}.
\end{align*}

\subsection{Cost Function}

The objective of the lead vehicle is to minimize the fuel consumption and control input, while maintaining a set reference velocity. The objective of the follower vehicles in addition is to follow the preceding vehicles velocity, while maintaining a set intermediate spacing. Hence, similar to what we presented for the continuous LQR in \cite{Alam11:Dec}, the weights for a $M$ HDV platoon can be set up based upon the performance objectives as
{\small{
\begin{align}
J(u^*)=&\min_{u}~\sum_{k=0}^{N-1}\Big(\sum_{i=2}^{M}w_i^{\tau}(d_{(i-1)i}(k)-\tau v_i(k))^2 \nonumber\\
&\hspace{13.5mm}+ w_i^{\Delta{v}}(v_{i-1}(k)-v_i(k))^2\nonumber\\
& \hspace{13.5mm}+w_i^{d}d_{(i-1)i}^2(k)+\sum_{i=1}^{M}w_i^{v}v_{i}^2(k)+w_i^{u_i}u_i^2(k)\Big) \nonumber\\
=&\min_{u}~\sum_{k=0}^{N-1}\sum_{i=2}^{M}\begin{bmatrix} v_{i-1}(k)\\
d_{(i-1)i}(k)\\
v_i(k)\end{bmatrix}^TQ_i\begin{bmatrix} v_{i-1}(k)\\
d_{(i-1)i}(k)\\
v_i(k)\end{bmatrix}+R_iu_i^2(k)\nonumber\\
&\hspace{13.5mm}+w^{v_1}v_{1}^2(k)+w^{u_1}u_1^2(k)
\label{eq:weights}
\end{align}
}}
\noindent where
{\small{
\begin{align*}
\begin{split}
&Q_i=\begin{bmatrix}
w_i^{\Delta{v}} & 0 & -w_i^{\Delta{v}} \\
0 & w_i^{d}+w_i^{\tau} & -\tau{w}_i^{\tau} \\
-w_i^{\Delta{v}} & -\tau{w}_i^{\tau} & \tau^2w_i^{\tau}+w_i^{\Delta{v}}+w_i^{v}
\end{bmatrix},\\
&Q_1=\begin{bmatrix} w^{v_1} & 0 \\ 0 & w^{u_1} \end{bmatrix}, R_i=w_i^{u_i}.
\end{split}
%\label{eq:weightCosts}
\end{align*}
}}

The weights in \eqref{eq:weights} give a direct interpretation of how to enforce the objectives for a vehicle traveling in a platoon. The value of $w_i^{\tau}$ determines the importance of not deviating from the desired time gap. Hence, a large $w_i^{\tau}$ puts emphasis on safety. $w_i^{\Delta{v}}$ creates a cost for deviating from the velocity of the preceding vehicle, and $w_i^{u_{i}}$ punishes the control effort which is proportional to the fuel consumption. The following terms, $w_i^{d}, w_i^{v}$, put a cost on the deviation from the linearized states. Note that the main objective is to maintain a set intermediate distance, while maintaining a fuel efficient behavior. Therefore, $w_i^{\tau}, w_i^{\Delta{v}}$ and $w_i^{u_{i}}$ must be set larger than the remaining weights.

\subsection{Problem Formulation}
We consider a HDV platooning scenario where each vehicle only receives information regarding the relative position and velocity of the immediate neighboring vehicles. The objective is to design a controller that can handle a two time step delay.

The aim is to utilize the given structure of the considered system, where we want to minimize the cost function

\begin{align}
\begin{split}
J=&\textbf{E}\{x(N)^{T}Q_{0}x(N)\}\\
&+\sum_{k=0}^{N-1}\textbf{E}\{x^T(k)Qx(k)+u^T(k)Ru(k)\},
\end{split}
\label{cost}
\end{align}

\noindent subject to the sparse system dynamics in \eqref{mainsystem} and the performance objectives in \eqref{eq:weights}. The primary difficulty arises from the imposed information constraints given in \eqref{historyset}.

Thus, the problem that we solve in this paper is finding an analytical expression for an optimal control input $u_i(k)$, which must be a function of the admissible information set $\mathcal{I}_i(k)$, where each subsystem control input is unique and a linear function denoted as
\begin{align}
u_i(k)=\mu_i\bigl(\mathcal{I}_i(k)\bigr),\; i=1,\dots, M.
\label{control}
\end{align}

\begin{assumption}
The matrices $Q_0$ and $Q$ in \eqref{cost} are positive semi-definite, and $R$ is positive definite.
\label{costmatrix}
\end{assumption}

\section{Main Result}
In this section we present the optimal controller for three-vehicle problem. The proof for this result is presented in the remaining sections.
\begin{theorem}
\emph{
Suppose that $W$ is positive definite and that
Assumption $1$ holds. Define the matrix $D\triangleq\begin{bmatrix} F & M\end{bmatrix}$ where $M$ has the same sparsity structure as $G$. Let $S$ be the index set of non-zero elements of} $\text{vec}(D)$,
$$
S\triangleq\left\{i:\;\textrm{vec}_i(D) \neq 0\right\}.
$$
\emph{Suppose there exists a stabilizing solution $X$ to the algebraic Riccati equation
\begin{small}\begin{align*}
X=A^TXA+Q+A^TXB(B^T XB+R)^{-1}B^T XA
\end{align*}\end{small}
We then define
\begin{small}\begin{align*}
H&=B^TXB+R\vspace{0.1cm}\\
L&=(B^TXB+R)^{-1}B^T XA\vspace{0.1cm}\\
\end{align*}\end{small}
and let}
\begin{small}\begin{align*}
Y&=\begin{bmatrix} W\otimes (H+B^TL^THLB) & -W\otimes B^TL^TH\\ -W\otimes HLB & W\otimes H\end{bmatrix}\\
b&=\begin{bmatrix} W\otimes H \\ 0 \end{bmatrix}\mbox{vec}(L)+\begin{bmatrix} -W\otimes B^TL^TH \\ W\otimes H \end{bmatrix}\textrm{vec}(LA)
\end{align*}\end{small}
\emph{Then, the optimal controller gains are given by:}
\begin{small}\begin{align*}
\textrm{vec}^{\star}(F)&=\begin{bmatrix} I & 0\end{bmatrix} [Y]_{SS}^{-1}[b]_S \vspace{1mm}\\
\textrm{vec}^{\star}(M)&=\begin{bmatrix} 0 & I\end{bmatrix} [Y]_{SS}^{-1}[b]_S
\end{align*}\end{small}
\emph{and the optimal controller has the realization
\begin{small}\begin{align*}
\zeta(k+1)=&Ax(k)+Bu(k)\\
\xi(k+1)=&A\zeta(k)+BM(x(k-1)-\zeta(k-1))+BL\xi(k)\\
u(k)=&F(x(k)-\zeta(k))\\
&+M(x(k-1)-\zeta(k-1))+L\xi(k)
\end{align*}\end{small}
}
\end{theorem}

Note that blocks of matrices $F$ and $M$ can be computed from the $\textrm{vec}^{\star}(F)$ and $\textrm{vec}^{\star}(G)$, respectively. For example, $\textrm{vec}^{\star}(F)=\textrm{vec}\left(\begin{bmatrix} F_{11} & F_{22} & F_{33}\end{bmatrix}\right)$.
It will be shown that $\xi(k)$ is the minimum-mean square estimate of $x(k)$ given the common information $x(0:k-2)$; that is, $\xi(k)=\textbf{E}\{x(k)|x(0:k-2)\}$. Thus, the optimal controller of three-vehicle problem is the centralized $LQR$ controller
under the classical information structure with two-step delay plus correction terms based on the local information at time $k$.

\section{Optimal Controller Derivation}
\label{sec:OptimalController}
In this section, we present the preliminary lemmas that are used to prove the results in Theorem $1$.
Before proceeding further, we need to state the following proposition which later permits us to decompose
$J$ into two separate parts.
\begin{proposition}[\cite {Astrom:70}]\emph{
Define the matrices
\begin{small}
\begin{align}
X(k)=&A^TX^+A+Q\label{lqg}\\
&-(A^TX^+B)(B^T X^+ B+R)^{-1}(B^T X^+A)\nonumber\\
H(k)=&B^T X^+ B+R\vspace{0.1cm}\nonumber\\
L(k)=&(B^T X^+ B+R)^{-1}B^T X^+ A\vspace{0.1cm}\nonumber
\end{align}
\end{small}
\noindent for $k=0,\cdots,N-1$ and where $X(N)=Q_0$. Then the cost function \eqref{cost} can be written as
\begin{small}
\begin{align}
J=&\underbrace{\sum_{k=0}^{N-1}\textbf{E}\left\{\bigl(u(k)-L(k)x(k)\bigr)^TH(k)\bigl(u(k)-L(k)x(k)\bigr)\right\}}_{J_u}\nonumber\\
&+\underbrace{x^T(0)X(0)x(0)+\sum_{k=0}^{N-1}\textbf{Tr}\{X(k+1)W\}}_{J_w}\nonumber
\label{lcost}
\end{align}
\end{small}
\noindent where both the zero-mean property of $w(k)$ and independence of $w(k)$ and $(x(k),u(k))$ are exploited.
Moreover, $J_{w}$ is independent of $u$.}
\label{dec}
\end{proposition}

From Proposition \ref{dec}, it can be seen that minimizing $J$ is equivalent to minimizing $J_{u}$. Note that, under the {Assumption~\ref{costmatrix}}, $H(k)$ is positive definite.

\subsection{State Decomposition}
The first step towards finding the optimal controller is decomposing the state vector into independent terms.
\begin{lemma}\label{lemma1}\emph{ The state vector can be decomposed as
\begin{align*}
x(k)=&\underbrace{w(k-1)+\bigl(A+BF(k-1)\bigr)w(k-2)}_{x^1(k)}\\
&+\underbrace{\textbf{E}\left\{x(k)|x(0:k-2)\right\}}_{x^2(k)}
\end{align*}
where $x^1(k)$ and $x^2(k)$ are independent random variables.}
\end{lemma}
\begin{proof}
The term $x^2(k)$ is the conditional estimate of state $x(k)$ given the piece of information shared between all vehicles, and $x^1(k)$ is the estimation
error. The independence between $x(k)-x^2(k)$ and $x^2(k)$ can be established by Proposition $4$b in the appendix. To calculate $x^1(k)$, we proceed in three steps. First consider
$$x(k-1)=Ax(k-2)+Bu(k-2)+w(k-2)$$
Since $x(k-2)$ belongs to the sequence $x(0:k-2)$ and $u(k-2)$ is a deterministic function of $x(0:k-2)$, we have
\begin{equation}
x(k-1)-\textbf{E}\{x(k-1)|x(0:k-2)\}=w(k-2)\label{300}
\end{equation}
where we used the zero-mean and independence of $w(k-2)$ and $x(0:k-2)$. The structure of controller is given by equation~\eqref{ucontrol}, so
$u(k-1)$ can be written as
\begin{small}
$$
u(k-1)=F(k-1) x(k-1)+G(k-1){x(k-2)}+f\bigl({x}(0:k-3)\bigr)
$$
\end{small}
Since $G(k-1){x(k-2)}+f\bigl({x}(0:k-3)$ is a deterministic function of $x(0:k-2)$, we have
\begin{eqnarray}
&&\hspace{-1cm}u(k-1)-\textbf{E}\{u(k-1)|x(0:k-2)\}\nonumber\\
&=&F(k-1)\bigl(x(k-1)-\textbf{E}\{x(k-1)|x(0:k-2)\}\bigr)\nonumber\\
&=&F(k-1)w(k-2)\label{301}
\end{eqnarray}
where we substituted ~\eqref{300} into the second line. Finally, note that $w(k-1)$ and $x(0:k-2)$ are independent. Therefore,
\begin{align}
x(k)&-\textbf{E}\{x(k)|x(0:k-2)\}\notag\\
=&w(k-1)+A\left(x(k-1)-\textbf{E}\{{x}(k-1)|x(0:k-2)\}\right)\nonumber\\
&\hspace{1.35cm}+B\left(u(k-1)-\textbf{E}\{{u}(k-1)|x(0:k-2)\}\right)\nonumber\\
=&w(k-1)+\bigl(A+BF(k-1)\bigr)w(k-2)\label{32}
\end{align}
where we substituted ~\eqref{300} into the second line and~\eqref{301} into the third line. Thus the result follows.
\end{proof}
\subsection{Controller Decomposition}
Now that the state has been decomposed into two independent terms, the control input $u(k)$ can be decomposed in a similar fashion.
\begin{lemma}\label{lemma2} \emph{The control input $u(k)$ can be decomposed as
\begin{align*}
u(k)&=\underbrace{F(k)w(k-1)+M(k)w(k-2)}_{u^1(k)}+u^2(k)
\end{align*}
where $u^1(k)$ and $u^2(k)$ are independent, $u^2(k)$ is a linear function of $x(0:k-2)$, and
$$M(k)=F(k)\left(A+BF(k-1)\right)+G(k).$$}
\end{lemma}
\begin{proof}
Let $u^2(k)=\textbf{E}\{u(k)|x(0:k-2)\}$, then $u^2(k)$ is a linear function of $x(0:k-2)$ and independent of $u(k)-u^2(k)$. Note that $f(x(0:k-2))$ is a linear function of $x(0:k-2)$, so $u^1(k)$ is computed as
\begin{align*}
u^1(k)=&u(k)-\textbf{E}\{u(k)|x(0:k-2)\}\notag\\
=&F(k)\bigl(x(k)-\textbf{E}\{x(k)|x(0:k-2)\}\bigr)\nonumber\\
&+G(k)\bigl(x(k-1)-\textbf{E}\{x(k-1)|x(0:k-2)\}\bigr)\nonumber\\
=&F(k)(w(k-1)+\left(A+BF(k-1)\right)w(k-2))\\
&+G(k)w(k-2)
\end{align*}
where we used equation~\eqref{ucontrol} in the first line, ~\eqref{32} in the second line and~\eqref{300} in the third line. The proof is completed by defining $M(k)=F(k)\left(A+BF(k-1)\right)+G(k)$.
\end{proof}
\begin{remark}
Since $B$ and $F$ are diagonal matrices, $G(k)$ and $F(k)A$ have the same sparsity structures. Therefore, sparsity structure of $M(k)$ and $G(k)$ are also the same.
\end{remark}

From Lemmas \ref{lemma1} and \ref{lemma2}, $x^2(k)$ and $u^2(k)$ are linear functions of $x(0:k-2)$ which is independent of $x^1(k)$ and $u^1(k)$. As a result the cost function $J_{u}$ can be decomposed as:
\begin{small}
\begin{align*}
J_{u}&=\underbrace{\sum_{k=0}^{N-1}\textbf{E}\left\{\bigl(u^1(k)-L(k)x^1(k)\bigr)^T H(k)\bigl(u^1(k)-L(k)x^1(k)\bigr)\right\}}_{J^1_{u}}\nonumber\\
&+\underbrace{\sum_{k=0}^{N-1}\textbf{E}\left\{\bigl(u^2(k)-L(k)x^2(k)\bigr)^T H(k)\bigl(u^2(k)-L(k)x^2(k)\bigr)\right\}}_{J^2_{u}}\nonumber\\
\end{align*}
\end{small}
The advantage of this decomposition of $J_{u}$ is that we now have two subproblems on the form
\begin{align}
\begin{split}
\min_{u^1}~&J^1_{u}(x^1,u^1)\\
\mbox{subject to}~&u^1(k)=F(k) w(k-1)+M(k) w(k-2)
\end{split}
\label{decomopt1}
\end{align}
\begin{align}
&\hspace{-23mm}\min_{u^2}~J^2_{u}(x^2,u^2)\nonumber\\
&\hspace{-30mm}\mbox{subject to}~u^2(k)=f\bigl(x(0:k-2)\bigr)
\label{decomopt2}
\end{align}
\subsection{Finite Horizon Controller Derivation}
First consider minimization problem (\ref{decomopt2}). Before proceeding, let us state the following proposition which allows us to find the optimal control $u^2(k)$.
\begin{proposition}[\cite {Astrom:70}]\emph{
Consider the discrete time linear system
$$
x(k+1)=Ax(k)+B u(k)+w(k)
$$
where $w(k)$ is a zero mean Gaussian white noise. Assume that $u(k)=\mu\bigl(x(0:k)\bigr)$. Then the optimal control which minimizes the cost function $J_u$, is given by
$$
u(k)=L(k) x(k)
$$}
\label{clqg}
\end{proposition}

The mapping from $x^2(k)$ to $u^2(k)$ is given in the following lemma.
\begin{lemma}\label{lemma3}\emph{
The dynamics of $x^2$ can be written as
\begin{equation}
x^2(k+1)=Ax^2(k)+Bu^2(k)+T(k)w(k-2)
\label{510}
\end{equation}
where $T(k)=A(A+BF(k-1))+BM(k)$.}
\end{lemma}
\begin{proof}
See appendix.
\end{proof}

The following theorem shows that $u^2(k)$ is exactly the optimal controller for centralized information structure with two step delay, where
the information set of each vehicle is $\mathcal{I}_{i}(k)=\left\{x(0:k-2)\right\}$.
\begin{theorem}\emph{
Given that Assumption $1$ holds, an optimal solution to (\ref{decomopt2}) is given by
\begin{align}
u^2(k)=L(k)x^2(k)
\label{500}
\end{align}}
\end{theorem}
\begin{proof}
Consider the system (\ref{510}) together with the cost function $J^2_{u}$. Both $x(0:k-2)$ and $u^2(k)$ are linear functions of $x(0:k-2)$ which is independent of $w(k-2)$. Hence, finding the optimal control $u^2(k)$ is now a centralized LQR problem. Applying proposition \ref{clqg}, we obtain (\ref{500}).
\end{proof}

We now turn to the optimization problem (\ref{decomopt1}). Recalling the expansions of $x^1(k)$ and $u^1(k)$ in terms of
$w(k-1)$ and $w(k-2)$, the expected value of the $k^{th}$ term of $J^1_{u}$ can be expanded as follows:
\begin{small}
\begin{align*}
\textbf{E}\{&\bigl(u^1(k)-L(k)x^1(k)\bigr)^T H(k)\bigl(u^1(k)-L(k)x^1(k)\bigr)\}\\
&=\textbf{Tr}\{H(k)(F(k)-L(k))W(F(k)-L(k))^T\}\nonumber\\
&\hspace{3mm}+\textbf{Tr}\{H(k)\bigl(M(k)-L(k)(A+BF(k-1))\bigr)W\\
&\hspace{2.4cm}\times \bigl(M(k)-L(k)(A+BF(k-1))\bigr)^T\}
\end{align*}
\end{small}
where we used Proposition $4$a in the appendix and the fact that $w(k-1)$ and $w(k-2)$ are independent. To minimize $J^1_{u}$ with respect to $F(0),\ldots,F(k)$ and $M(1),\ldots,M(k)$, the difficulty is that $F$ and $M$ must satisfy specified sparsity constraints. We use vectorization of matrices to simplify our optimization problem.

Let us define the matrix $D(k)$ as follows
\begin{small}
\begin{align*}
D(k)\triangleq\begin{bmatrix} F(k-1) & M(k)\end{bmatrix} \;\;\in\mathbb{R}^{m\times 2p},\;k=1,\ldots,N-1
\end{align*}
\end{small}
and $D(N)\triangleq F(N-1)$. Then $\textrm{vec}\bigl(D(k)\bigr)$ is given by
$$\begin{bmatrix}\textrm{vec}\bigl(F(k-1)\bigr) \\ \textrm{vec}\bigl(M(k)\bigr)\end{bmatrix}\;\;\; \in\mathbb{R}^{2mp},\;k=1,\ldots,N-1$$
and $\textrm{vec}(D(N)\bigr)=\textrm{vec}\bigl(F(N-1)\bigr)$.
Because of the specified sparsity of $F$ and $M$, some components of $\textrm{vec}\bigl(D(k)\bigr)$ must be zero. Let $S$ be the index set of non-zero elements of
$\textrm{vec}\bigl(D(k)\bigr)$, i.e.
$$
S\triangleq\left\{i:\;\textrm{vec}_i\bigl(D(k)\bigr) \neq 0\right\}
$$
Note that $\textrm{vec}\bigl(D(k)\bigr)$ and $\textrm{vec}^{\star}\bigl(D(k)\bigr)$ are related by nonsquare matrix. We define this matrix to be $E$, where dimensions implied by the context, so that $\textrm{vec}\bigl(D(k)\bigr)=E\textrm{vec}^{\star}\bigl(D(k)\bigr)$. The columns of $E$ are ${e_j}$ for ${j\in S}$ where $e_j$ denotes a column vector having all zeros except a $1$ at the $j^{th}$ position. Since exactly one entry in each column of $E$ is equal to $1$, $E^T X E$ is a
sub-matrix of $X$ containing exactly those rows and columns corresponding to the set $S$.
We illustrate the above definition via an example. Let $D=\text{diag}(d_{11},d_{22},d_{33})\in\mathbb{R}^{3\times 3}$. For this matrix, $S=\{1,5,9\}$, $E=\begin{bmatrix} e_1& e_5 & e_9\end{bmatrix}$, and
$\textrm{vec}^{\star}(D)=[d_{11},d_{22},d_{33}]^T$.

In the following lemma, we show that a vectorization of matrices $F$ and $M$ makes the cost function $J^1_{u}$ a
sum of quadratic functions without constraints.

\begin{lemma}\emph{
Define
\begin{small}
\begin{align*}
Y_{11}(k)&=W\otimes (H(k-1)+B^TL^T(k)H(k)L(k)B)\\
Y_{12}(k)&=-W\otimes B^TL^T(k)H(k)\\
Y_{22}(k)&=W\otimes H(k)
\end{align*}
\end{small}
and let}
\begin{small}
\begin{align*}
Y_k&=\begin{bmatrix} Y_{11}(k) & Y_{12}(k)\\ Y_{12}^T(k) & Y_{22}(k)\end{bmatrix} \\
b_k&=\begin{bmatrix} Y_{22}(k-1) \\ 0 \end{bmatrix}\textrm{vec}\left(L(k-1)\right)+\begin{bmatrix} Y_{12}(k) \\ Y_{22}(k) \end{bmatrix}\textrm{vec}\left(L(k)A\right)
\end{align*}
\end{small}
\emph{for $k=1,\ldots,N-1$, and}
\begin{small}
\begin{align*}
Y_N&=W \otimes H(N-1)\\
b_N&=\bigl(W \otimes H(N-1)\bigr)\textrm{vec}\left(L(N-1)\right)
\end{align*}
\end{small}
\emph{Then optimization problem (\ref{decomopt1}) is equivalent to}
\begin{eqnarray}
\min_{\textrm{vec}^{\star}(D(k))} \sum_{k=1}^{N} & \frac{1}{2}\textrm{vec}^{\star}(D(k))^T [Y_k]_{SS} \textrm{vec}^{\star}(D(k))\label{800}\\
&-\textrm{vec}^{\star}(D(k))^T [b_k]_S\nonumber
\end{eqnarray}
\emph{Moreover, $Y_k$ is positive definite.}
\end{lemma}
\begin{proof} See appendix.
\end{proof}

The advantage of this equivalent reformulation of the problem is that we have $N$ quadratic functions without constraints and thus the optimal
controller gains can be computed by simply minimizing these functions separately.
\begin{theorem}\emph{
Suppose $W$ is positive definite and Assumption $1$ holds. Then the optimal gains of controllers are given by:}
\begin{small}
\begin{align*}
\textrm{vec}^{\star}(F(k-1))&=\begin{bmatrix} I & 0\end{bmatrix} \textrm{vec}^{\star}(D(k))\\
\textrm{vec}^{\star}(M(k))&=\begin{bmatrix} 0 & I\end{bmatrix} \textrm{vec}^{\star}(D(k))
\end{align*}
\end{small}
\emph{for $k=1,\ldots,N-1$ and }$\textrm{vec}^{\star}(F(N-1))=\textrm{vec}^{\star}(D(N))$, \emph{where}
$\textrm{vec}^{\star}(D(k))=[Y_k]_{SS}^{-1}[b_k]_S$.
\end{theorem}
\subsection{Steady State Controller Derivation}
Assume that the solution to algebraic Riccati equation (\ref{lqg}), $X(k)$, converges to
the stabilizing solution as $k$ approaches $\infty$:
\begin{align*}
X=A^TXA+Q+A^TXB(B^T XB+R)^{-1}B^T XA
\end{align*}
Since $H(k)$ and $L(k)$ are specified
by $X(k)$, they respectively converge to matrices $H$ and $L$ as follows:
\begin{small}
\begin{align*}
H=B^TXB+R,~L=(B^TXB+R)^{-1}B^T XA
\end{align*}
\end{small}
Then $Y_k$ and $b_k$ will approach the values of $Y$ and $b$ given by
\begin{small}
\begin{align*}
Y&=\begin{bmatrix} W\otimes (H+B^TL^THLB) & -W\otimes B^TL^TH\\ -W\otimes HLB & W\otimes H\end{bmatrix}\\
b&=\begin{bmatrix} W\otimes H \\ 0 \end{bmatrix}\textrm{vec}(L)+\begin{bmatrix} -W\otimes B^TL^TH \\ W\otimes H \end{bmatrix}\textrm{vec}(LA)
\end{align*}
\end{small}
Thus, the optimal gains are calculated to be
\begin{small}
\begin{align*}
\textrm{vec}^{\star}(F)=\begin{bmatrix} I & 0\end{bmatrix} [Y]_{SS}^{-1}[b]_S\\
\textrm{vec}^{\star}(M)=\begin{bmatrix} 0 & I\end{bmatrix} [Y]_{SS}^{-1}[b]_S
\end{align*}
\end{small}
\subsection{Estimation Structure}
Having determined the optimal controller, we turn now to analyze this result. Define $\zeta(k)=x(k)-w(k-1)$. Hence, we obtain the
following state-space system
$$
\zeta(k+1)=Ax(k)+Bu(k)
$$
with initial condition $\zeta(0)=0$. Note that the assumptions about the information structure and sparsity structure of $A$ and $B$ guarantee that each vehicle can update $\zeta(k)$ at time $k$. For example, consider Vehicle $1$. Since Vehicle $1$ has access to
$x_2(k-1)$ at time $k$, It can construct $\zeta_1(k)=A_{11}x_1(k-1)+A_{12}x_2(k-1)+B_1u_1(k-1)$.
Letting $\xi(k)=\textbf{E}\{x(k)|x(0:k-2)\}$ the optimal control policy can be written as
\begin{small}\begin{align*}
u(k)=F(x(k)-\zeta(k))+M(x(k-1)-\zeta(k-1))+L\xi(k)
\end{align*}\end{small}
In order to fully specify $u(k)$, the conditional estimates $\xi(k)$, as well as the matrices $L$, $F$ and $G$ must be computed. We have
\begin{small}
\begin{align*}
\xi(k+1)&=\textbf{E}\{x(k+1)\mid x(0:k-1)\}\nonumber\\
&\hspace{0cm}=A\textbf{E}\{x(k)\mid x(0:k-1)\}+B\textbf{E}\{u(k)\mid x(0:k-1)\}\\
&=A\zeta(k)+BM(x(k-1)-\zeta(k-1))+BL\xi(k)
\end{align*}
\end{small}
\section{Numerical Results}
\label{sec:Performance}
We evaluate the performance of the system with the controller by giving an example of a realistic scenario that HDV platoons often face on the road. In practice, varying traffic conditions often mandate a deviation in the lead vehicle's velocity. Therefore, integral action for the lead vehicle is added as a state to the system presented in \eqref{eq:plantModel}, to model such disturbances.

\begin{figure}[t]
\begin{center}
\includegraphics[width=7.4cm]{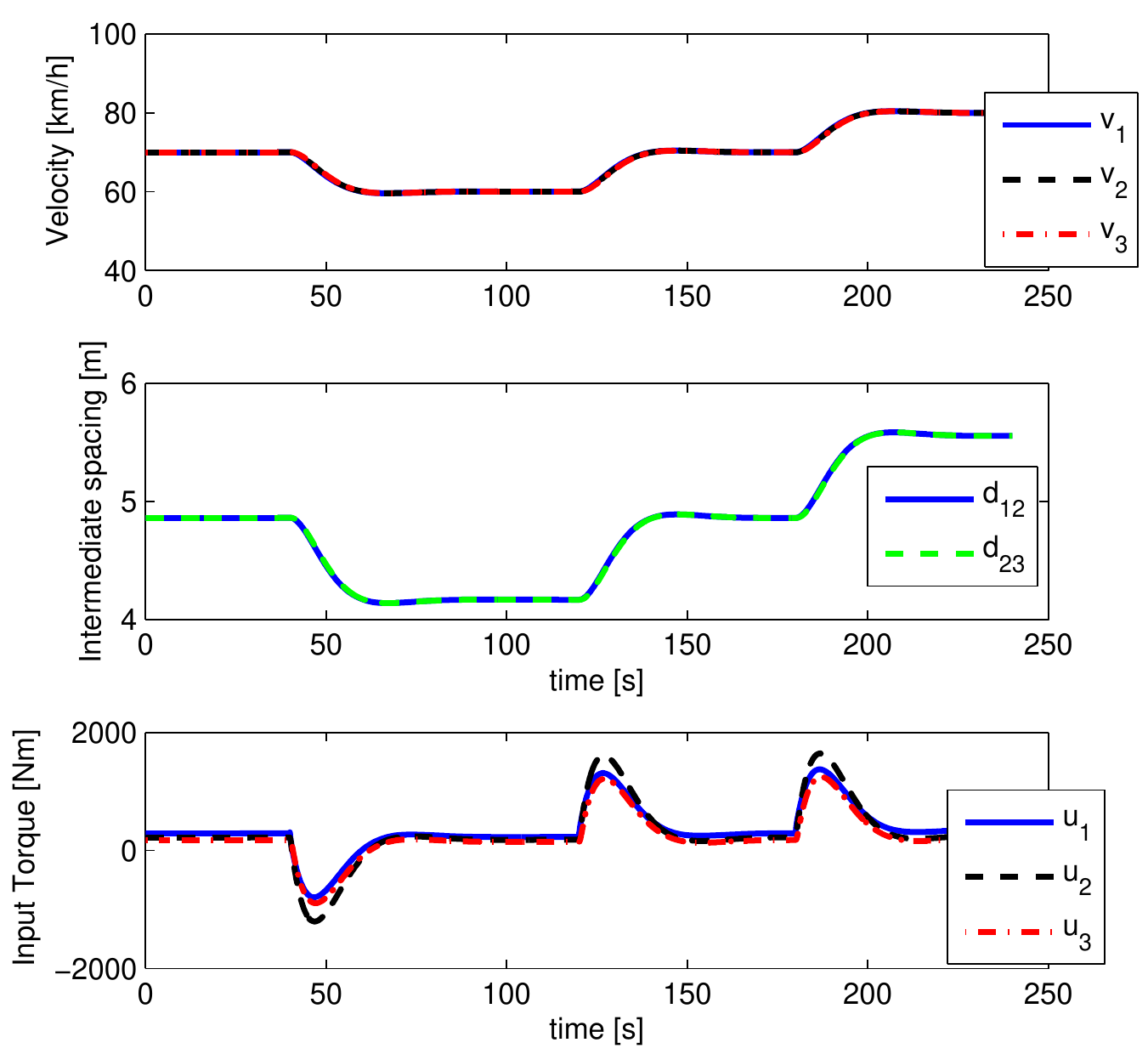}    % The printed column width is 8.4 cm.
\caption{Three HDV platoon, where a disturbance in velocity of the lead vehicle is imposed. The top plot shows the velocity trajectories for the $M=3$ HDV platoon, the bottom plot shows the intermediate spacings, and the bottom plot shows the control inputs. The trajectories are obtained through the optimal distributed control and subindexed $i$, where $i=1,2,3$ denote the platoon position index.}
\label{fig:disturbance}
\end{center}
\end{figure}

We consider a heterogeneous platoon, where the masses are set to $[m_1, m_2, m_3]=[30000, 40000, 30000]$\,kg. All the vehicles are assumed to be traveling in the steady state velocity  $v_0=19.44$\,m/s ($70$\,km/h) at time gap $\tau=0.25$\,s, which gives an intermediate distance of $d_0=4.86$. The maximum engine and braking torque for a commercial HDV varies based upon vehicle configuration but can be approximated to be 2500\,Nm and 60000\,Nm/Axle respectively.

State disturbances as well as several lead vehicle deviation disturbances are imposed on the system (Fig.~\ref{fig:disturbance}). The lead vehicle deviation disturbances can be explained by the following scenario. The platoon travels along a road where the road speed is 70\,km/h. Suddenly a slower vehicle enters the lane through a shoulder path at a lower speed. The lead vehicle must therefore reduce its speed to 60\,km/h, in turn forcing the follower vehicles to reduce their speed and adapt their relative distance accordingly. After a while, the slower vehicle increases its speed to the road speed of 70\,km/h and no longer inhibits the platoon. Hence, the lead vehicle again resumes the road speed and the follower vehicles adapt the speed and distance automatically as well. Finally, the platoon arrives at a point where the road speed is changed to 80\,km/h.

Fig.~\ref{fig:disturbance} shows the velocity trajectories of three HDV platoon in the top plot, the corresponding intermediate spacings in the middle plot, and the required control input to handle the disturbances in the bottom plot. The trajectories nearly lie on top of each other, showing that the proposed controller performs a tight control and the disturbances are handled well. There is no overshoot in the velocity or intermediate spacing tracking. Furthermore, the control input is well within the feasible physical range. The weight normalized control input energy required to handle the imposed disturbances is reduced by 15\,\% for the second vehicle and 14\,\% for the third vehicle, with respect to the first vehicle. Hence, the controller displays a fuel efficient behavior, since the input energy is directly proportional to the fuel consumption. The theoretical value, in this case, for the cost function with the proposed optimal distributed control is only 0.01\,\% higher than a fully centralized control with full state information at all times. On the other hand, the proposed controller produces a 67\,\% lower theoretical cost compared to a centralized control with two step time delays.

\section{SUMMARY AND CONCLUSIONS}
\label{sec:Concl}

We have presented an analytical controller, which is optimal under a delayed information sharing pattern for chain structures. A discrete time HDV platoon model has been derived that includes physical coupling with both neighboring vehicles. The results show that the cost function with proposed controller is very close to the fully centralized cost and better than the cost for the centralized case with two time delays. Hence, the cost function can be significantly reduced by considering additional available local information. The controller maintains a tight control even though time delays are imposed.

For future work, we plan to extend to the presented results to $M$ HDVs and arbitrary time step delays, which is relevant for HDV platooning.

%%%%%%%%%%%%%%%%%%%%%%%%%%%%%%%%%%%%%%%%%%%%%%%%%%%%%%%%%%%%%%%%%%%%%%%%%%%%%%%%

\bibliography{IEEEabrv,References}

\appendix
\subsection{Preliminaries}
\begin{proposition}[\cite{Horn:96}]\emph{ If $A$, $B$, $C$, $D$ and $X$ are suitably dimensioned matrices, then}
\begin{enumerate}
\item[a)] $\textrm{vec}(AXB)=(B^T\otimes A)\textrm{vec}(X)$,
\item[b)] $(A\otimes B)(C\otimes D)=(AC)\otimes (BD)$,
\item[c)] \emph{If $A$ and $B$ are positive definite, then so is $A\otimes B$,}
\item[d)] $\textbf{Tr}\{AXBX^T\}=\textrm{vec}^T(X)(B^T\otimes  A)\textrm{vec}(X)$,
\item[e)] $(A\otimes B)^{-1}=A^{-1}\otimes B^{-1}$.\emph{($A$ and $B$ are nonsingular)}
\end{enumerate}
\label{prop1}
\end{proposition}
\begin{proposition}[\cite{Astrom:70}]\emph{
Let $x$ and $y$ be zero-mean random vectors with a jointly Gaussian distribution. Assume $S$ be
a symmetric matrix. Then the following facts hold:
\begin{enumerate}
\item[a)] $\textbf{E}\{x^TSx\}=\textbf{Tr}\bigl\{S\textbf{{E}}\{xx^T\}\bigr\}$.
\item[b)] $\textbf{E}\{x|y\}$ and $x-\textbf{E}\{x|y\}$ are independent.
\end{enumerate}
\label{prop2}}
\end{proposition}
\begin{proposition}[\cite{Horn:96}]\emph{
Suppose that a symmetric matrix is partitioned as \begin{small}$\begin{bmatrix} A & B \\ B^T & C\end{bmatrix}$\end{small},
where $A$ and $C$ are square. This matrix is positive definite if and only if $C$ and $\bigtriangleup=A-BC^{-1}B^T$ are positive
definite.}
\end{proposition}
\subsection{Proof Lemma 3}
First note that $x^2(k)=x(k)-x^1(k)$. Thus,
\begin{align*}
x^2(k+1)=&Ax(k)+Bu(k)-\bigl(A+BF(k)\bigr)w(k-1)\\
=&Ax^1(k)+Ax^2(k)+Bu^1(k)+Bu^2(k)\\
&-\bigl(A+BF(k)\bigr)w(k-1)
\end{align*}
The proof is completed by substituting $x^1(k)=w(k-1)+\bigl(A+BF(k-1)\bigr)w(k-2)$ and $u^1(k)=F(k)w(k-1)+M(k)w(k-2)$ into the second line.
\subsection{Proof Lemma 4}
The equivalence of optimization problems (\ref{decomopt1}) and (\ref{800}) follows simply by vectorization of matrices.
First note that \begin{small}$\textrm{vec}\bigl(F(k-1)\bigr)=\begin{bmatrix} I & 0\end{bmatrix}\textrm{vec}\bigl(D(k)\bigr)$\end{small}. Thus
\begin{small}
\begin{align*}
\textrm{vec}(F^--L^-)&=[I\;0]\textrm{vec}\bigl(D(k)\bigr)-\textrm{vec}(L^-)
\end{align*}
\end{small}
From Propositions $3$b and $3$d, we have
\begin{small}
\begin{align*}
&\hspace{-1.5cm}\textbf{Tr}\left\{H^-(F^--L^-)W(F^--L^-)^T\right\}\\
=&\textrm{vec}^T\bigl(D(k)\bigr)\begin{bmatrix} W\otimes H^-& 0\\ 0 & 0\end{bmatrix}\textrm{vec}\bigl(D(k)\bigr)\\
&-2\textrm{vec}^T(L^-)\begin{bmatrix} W\otimes H^- & 0\end{bmatrix}\textrm{vec}\bigl(D(k)\bigr)\\
&+\textrm{vec}^T(L^-)(W\otimes H^-)\textrm{vec}(L^-)
\end{align*}
\end{small}
Likewise, $\textrm{vec}\bigl(M(k)\bigr)=\begin{bmatrix} 0 & I\end{bmatrix}\textrm{vec}\bigl(D(k)\bigr)$. Then
\begin{small}
\begin{align*}
&\textrm{vec}\bigl(M(k)-L(A+BF^-)\bigr)=\\
&\hspace{2cm}\bigl([0\;\;I]-[LB\;\;0]\bigr)\textrm{vec}\bigl(D(k))\bigr)-\textrm{vec}(LA)
\end{align*}
\end{small}
Therefore,
\begin{small}
\begin{align*}
&\hspace{-5mm}\textbf{Tr}\bigl\{H\bigl(M-L(A+BF^-)\bigr)W\bigl(M-L(A+BF^-)\bigr)^T\bigr\}\\
=&\textrm{vec}^T(D)\begin{bmatrix} W\otimes B^TL^THLB & -W\otimes B^TL^TH\\ -W\otimes HLB & W\otimes H\end{bmatrix}\textrm{vec}(D)\\
&-2\textrm{vec}^T(LA)\begin{bmatrix} -W\otimes HLB & W\otimes H\end{bmatrix}\textrm{vec}\bigl(D(k)\bigr)\\
&+\textrm{vec}^T(LA)(W\otimes H)\textrm{vec}(LA)
\end{align*}
\end{small}
After Substituting these values back into $J_{u}^1$, using $\textrm{vec}(D)=E\textrm{vec}^{\star}(D)$, and eliminating constant terms, we arrive at (\ref{800}).

The only part that remains to be proved is that $Y_k$ is positive definite. Since $W$ and $H(k)$ are positive definite, $Y_{22}(k)$ is positive definite according to
Proposition $3$c. Proposition~$3$e then implies that $Y_{22}^{-1}(k)=W^{-1}\otimes H^{-1}(k)$. From Proposition $3$b, we have
\begin{align*}
Y_{12}(k)Y_{22}^{-1}(k)Y^T_{12}(k)=W\otimes B^TL^T(k)H(k)L(k)B
\end{align*}
Consequently,
\begin{align*}
\bigtriangleup(k)&=Y_{11}(k)-Y_{12}(k)Y_{22}^{-1}(k)Y^T_{12}(k)\\
&=W\otimes H(k-1)
\end{align*}
Since $\bigtriangleup(k)$ and $Y_{22}(k)$ are positive definite, from Proposition $5$, $Y_k$ is positive definite. Finally note that $E^T$ has full row rank, so $[Y]_{SS}=E^TY_kE$ is positive definite.

\end{document}